\documentclass[letterpaper, 10 pt, conference]{ieeeconf}
\usepackage{comment}
\usepackage{amssymb, amsmath,graphicx,charter, latexsym}
\usepackage{amsfonts}
\usepackage{bbm}
\usepackage{algorithm}
\newtheorem{definition}{Definition}
\newtheorem{lemma}{Lemma}

\newtheorem{corollary}{Corollary}
\newtheorem{theorem}{Theorem}

\newtheorem{assumption}{Assumption}

\usepackage{graphicx}
\usepackage{url}
\usepackage{epstopdf}
\usepackage{gensymb}
\usepackage{caption,xcolor}
\usepackage{subcaption}

\IEEEoverridecommandlockouts     

\title{Optimal Information Updating based on Value of Information }
\author{Rahul Singh, Gopal Krishna Kamath and P. R. Kumar
\thanks{G. K. Kamath and P. R. Kumar are with Department of ECEN, Texas A\&M University}}
\begin{document}
\maketitle
\section{Abstract}
We address the problem of how to optimally schedule data packets over an unreliable channel in order to minimize the estimation error of a simple-to-implement remote linear estimator using a constant ``Kalman'' gain to track the state of a Gauss Markov process. The remote estimator receives time-stamped data packets which contain noisy observations of the process. Additionally, they also contain the information about the ``quality'' of the sensor\slash source, i.e., the variance of the observation noise that was used to generate the packet. In order to minimize the estimation error, the scheduler needs to use both while prioritizing packet transmissions. It is shown that a simple index rule that calculates the \emph{value of information} (VoI) of each packet, and then schedules the packet with the largest current value of VoI, is optimal. The VoI of a packet decreases with its age, and increases with the precision of the source. Thus, we conclude that, for constant filter gains, a policy which minimizes the age of information does not necessarily maximize the estimator performance.
\section{Introduction}\label{sec:intro}
There is a growing interest in providing real-time status updates in order to serve applications that depend upon the freshness of information available to them. For example, timely weather updates, stock prices information, Internet of Things devices, etc. Thus, the problem of ensuring timely status
updates in real-time applications has received much attention ~\cite{kaul2012real,DBLP:conf/allerton/KadotaUSM16,DBLP:journals/ton/KadotaSUSM18}. The ``Age of Information'' (AoI), captures the freshness of the information that is available with the end application, 
and recent works have designed scheduling policies\slash network controllers with the objective of  minimizing the AoI in scenarios where multiple applications share a common network infrastructure. For example, more recently generated data packets are prioritized over older packets~\cite{DBLP:conf/isit/BedewySS17}, or the packet of a user that currently has a larger value of age at the destination is prioritized over a user having smaller age, etc.

However, in many applications, it is not only the freshness of a packet that matters, but also the content that it contains.
We will be interested in applications that generate an  estimate of a process after receiving data packets that contain information about this process. Since the packets contain only noisy observations of the process, and not the ``true'' value of the process, it is also important that a packet which contains measurements that have a higher precision\slash low noise, must be prioritized over a packet that has the same age but has a lower precision measurement. 

In summary, the quality of a packet is judged not only on the basis of how ``fresh'' the packet is, i.e., its age or the time since it was generated at the source, but it also depends upon the following two factors 
\begin{enumerate}
\item How much information it contain about the process that is being monitored,
\item How important this information is to the algorithm that is being used to update the status at the destination node. 
\end{enumerate}
Regarding 1), we note that the information content in a packet is described by the joint probability distribution of its content and the true state\slash status of the process. In order for the factor 2) above to become crucial, it becomes important that the algorithm that is being deployed by the application uses the information contained in these packets efficiently. 

In view of the above discussion, the following question arises naturally: How do we prioritize packets for scheduling in order to optimize the performance of such real-time systems? Should the packets be prioritized according to their age, or the noisiness of the observation contained in them, or a combination of both? 
In this work, we provide a concrete answer to this question when the process of interest that is being monitored by an application is Gauss Markov, while the algorithm that is being used to generate an estimate of the process is a 
simple-to-implement linear filter
with a constant gain. We show that the optimal scheduler takes the following form: it attaches an index to each packet that is present in its queue, and then schedules the packet having the largest value of the index. This index depends upon the following two factors a) the age of the packet, b) the mutual information between the packet content, and the system status, or equivalently the variance of the noise associated with the measurement. 

\begin{figure}[!h]
	\centering
	\includegraphics[width=0.5\textwidth]{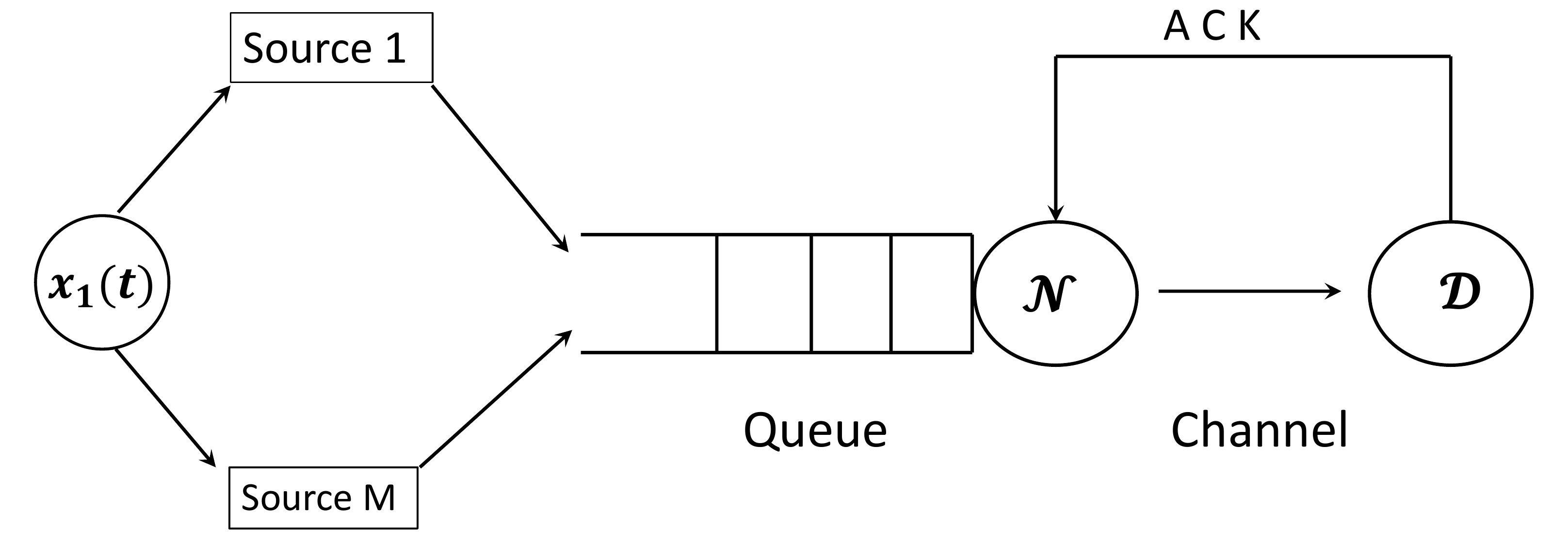}
	\caption{$N$ sources generate data packets which contain information about the process to the network node. The packets at the network node are queued in a single queue, and wait to be transmitted to the destination. The sources differ in the precision of their measurements of the process $x(t)$.
	}
	\label{fig:estimator}
\end{figure}

\section{Related Works}
Applications involving real-time systems such as the smart-grid, weather monitoring systems, etc., often require that the destination node, which comprises the end decision maker, has a constant availability of fresh \emph{staus updates} about the process of interest. The age-of-information, or simply age, performance metric was introduced in~\cite{kaul2012real,adelberg1995applying,Cho00synchronizinga,DBLP:conf/icde/GolabJS09}, and is an appropriate choice in order to encapsulate the notion of freshness of data. The age at a given time is equal to the time elapsed since the freshest data packet available to destination was generated. Existing works on age derive age-optimal policies for various kinds of networks, e.g., when multiple sources share a common queue~\cite{2012ISIT-YatesKaul}, broadcast networks~\cite{DBLP:conf/allerton/KadotaUSM16,DBLP:journals/ton/KadotaSUSM18}, under battery energy constraints on the packet transmitter~\cite{DBLP:conf/isit/Yates15}, and for wireless networks~\cite{DBLP:journals/iotj/JiangKZZN19,DBLP:conf/infocom/JiangZN019}. The work~\cite{DBLP:conf/isit/SunPU17} shows that age optimal policies minimize estimation error when the process of interest that is being monitored by the destination is a Weiner process. 

There are also other performance metrics relevant for real-time systems, which are 
also concerned with the timeliness of information delivery to the destination. In applications where the end-to-end packet delays are required to be small, the notion of ``timely throughput'', i.e., the throughput of packets that reach their destination before their deadlines, is useful. The works~\cite{cdcdelay,singh2018throughput} develop decentralized policies that maximize the timely throughput. Another useful metric is that of variance\slash variations in the packet interdelivery times. Minimizing this quantity ensures that the stream of packet deliveries as viewed by the end application looks close to the periodic delivery case when it regularly gets a packet every 1\slash throughput time-slots. See~\cite{guo2018risk,guosingh,singh2015index} for more details. Service smoothness or regularity is a useful concept in applications which require that the end application should not suffer from the problem that there are long durations in which they obtain no packet at all. An important example is video streaming which suffers from outages whenever the video buffer is emptied by non-delivery of packets for a long time~\cite{singh2015optimizing}. A scheduler that maintains service smoothness avoids such service starvation in cases when multiple applications share a common network resource~\cite{singh2015maxweight,rs,atilla1,atilla2}.

One may note that the performance metrics mentioned above try to connect either the freshness or the timeliness of information updates to the performance of the monitoring algorithm that is being deployed by the destination. However, the quality
of the measurements in all data packets need not be the same, that is to say that they differ in the amount of noise in the observations contained in them. Thus, the above performance metrics are unable to appropriately prioritize packet transmissions when the packets are generated by different sources that have varying precision. In this work, we address this problem for the case when the process of interest being monitored is a Gauss Markov process~\cite{kumar}, and the destination node employs a constant-gain linear filter to monitor the process. However, in this work, we do not consider more complicated signaling mechanisms such as the information conveyed by the non-arrival of packets. 

\section{Setup}
Throughout, we assume that time has been discretized into time-slots, and is numbered as $t=1,2,\ldots$. We assume that the system operates for $T$ time-slots, so that $t\in \left\{1,2,\ldots,T \right\}$. We denote the segment 
$\left\{\ell,\ell+1,\ldots,m\right\}$ of $m-\ell+1$ consecutive time-slots by $[\ell,m]$.

\emph{Process and Network Description}:
We consider the system shown in Fig.~\ref{fig:estimator}. There is a single network node $\mathcal{N}$ that receives packets containing information about the process $x(t),t\in [1,T]$ from $M$ sources. The process $x(t)\in\mathbb{R}$ is a scalar Gauss Markov process, and evolves as follows
\begin{align}\label{eq:single_system}
x(t+1) = ax(t) + w(t), t\in [1,T],
\end{align}
where $a$ describes its dynamics, while the noise $w(t)$ is a zero mean Gaussian process that is i.i.d. across time. There are $M$ sources that are equipped with sensors that take observations of the process $x(t)$. Each such source is characterized by the noisiness of its measurements, as described below. 

For a packet $\psi$, we let $t_{\psi}$ be the time at which the packet was generated at its source. Thus, the observation contained in $\psi$, denoted $y_{\psi}$, is given by,
\begin{align*}
y_{\psi} = x(t_\psi) +\sigma_{\psi} w_{\psi},
\end{align*}
where $w_{\psi}$ is a zero mean Gaussian random variable and $\sigma_{\psi}$ is the variance associated with observation noise corresponding to the source that generated $\psi$. Thus, the quantity $\sigma_{\psi}$ can assume one of the $M$ possible values. The observation noise $w_{\psi}$ is assumed to be independent across packets, and Gaussian with mean zero, variance $1$. We let $\tau_{\psi}(t)=t-t_\psi$ be the age of packet $\psi$  at time $t$. 

\emph{Network Node}: The $M$ sources provide their packets to a shared network node $\mathcal{N}$. $\mathcal{N}$ maintains a queue with infinite buffer size for storing the incoming data packets and is connected to the destination node $\mathcal{D}$ via an \emph{unreliable} link $\ell$. We denote the reliability of $\ell$ by $p$, the probability that an attempted packet transmission is successful. The link $\ell$ has a capacity of $1$ packet\slash time-slot, i.e., at most one packet can be transmitted on it in one time-slot. As shown in Fig.~\ref{fig:estimator}, we assume that upon a successful transmission, the receiver sends an acknolwedgement (ACK) to $\mathcal{N}$, so that the network node knows whether or not the packet transmission was successful \emph{after} having attempted the transmission at $t$. If the transmission is successful, it removes the delivered packet from its queue. We assume that the packet arrivals to $\mathcal{N}$ from the $M$ sources are governed by a well-defined stochastic process. We let $Q(t)$ denote the set of packets contained in the queue at time $t$.

\emph{Filtering}: The destination node has a linear filter~\cite{kumar} that produces real-time\slash online estimates $\hat{x}(t)$ of $x(t)$. Let $\psi(t)$ denote the packet that is delivered to the destination at time $t$.
Each such delivered packet is used for updating the estimate $\hat{x}(t)$, as follows:
\begin{small}
\begin{align}\label{eq:kalman_filter}
\hat{x}(t+1) =
\begin{cases}
a \hat{x}(t) + K \left(  a^{ \tau_{\psi(t)}(t) }y_{\psi(t)} - a \hat{x}(t) \right),  \mbox{ if } \psi(t) \neq \left\{\emptyset \right\}\\
a \hat{x}(t) \mbox{ if } \psi(t) = \left\{\emptyset \right\}, 
\end{cases}
\end{align}
\end{small}
\noindent where the scalar gain $K$ is a constant gain,
$\tau_{\psi(t)}(t)$ is the age of packet $\psi(t)$, $y_{\psi(t)}$ is the observation value contained in the packet $\psi(t)$, and we let $\left\{ \psi(t) = \left\{\emptyset \right\}  \right\}$ denote that no packet was delivered to $\mathcal{D}$ at $t$. 
Such an event can occur either because the queue $Q(t)$ was empty, or because the packet that was transmitted was dropped by the link $\ell$.
For simplicity, we use a constant gain $K$ for all time-slots $t$. The updates~\eqref{eq:kalman_filter} correspond to Kalman filtering with constant gain when the filter obtains delayed time-stamped packets, and it can utilize the values of delays in order to minimize the estimation error~\cite{DBLP:journals/automatica/ShiXM09}. See Theorem~3.1 and Fig.~2 of~\cite{DBLP:journals/automatica/ShiXM09} for more details.
Let 
\begin{align}\label{eq:error_def}
e(t):= x(t)-\hat{x}(t)
\end{align} 
denote the estimation error of the filter at time $t$.

We are interested in designing scheduling policies that can be implemented at the network node in order to minimize
the expected value of the cumulative squared estimation error over the time interval $[1,T]$, i.e., the quantity
\begin{align}\label{eq:error_cost}
\mathbb{E}\left(\sum_{t=1}^{T} e^2(t)\right),
\end{align}
where the expectation is taken with respect to the randomness of packet arrivals, observation and process noises $w_\psi,w(t),t\in[1,T]$ and the scheduling policy employed by $\mathcal{N}$. We denote a generic scheduling policy by $\pi$, and the packet chosen for transmission at $t$ by $u(t)$. A policy $\pi$ makes the decision $u(t)$ denoting which packet to transmit, on the basis of its past observations, i.e., $\left\{Q(s)\right\}_{s=1}^{t},\left\{u(s),x(s)\right\}_{s=1}^{t-1},$ where $Q(t)$ denotes the queue length at time $t$. 
At any time during the operation, the scheduler can access the age $\tau_{\psi}$, quality
$\sigma_{\psi}$, or observation $y_{\psi}$ of each packet $\psi$ present in its queue. Thus, the objective can be stated as follows,
\begin{align}\label{eq:cost}
\min_{\pi}J(\pi) : = \mathbb{E}_{\pi}\left(\sum_{t=1}^{T} e^2(t)\right).
\end{align}
\section{An Index Based Policy}
We begin by deriving some preliminary results concerning the estimation error process $e(t)$, and also introduce a simple-to-implement policy that is optimal when the scheduler has to make decision only at the last operating time $t=T-1$. We then use backward recursion argument to obtain the scheduler decision for all $t.$ 

We obtain the following expression for the evolution of the error $e(t)$ at the destination node.
Consider the linear
 filter~\eqref{eq:kalman_filter} that performs updates by utilizing the observation $y_{\psi(t)}$ contained in packet $\psi(t)$ that was delivered to it at time $t$. Let $a_c := a(1-K)$ denote the closed-loop estimator 
 gain of the linear filter. Then, its estimation error $e(t)= x(t)-\hat{x}(t)$ evolves as follows,
\begin{align}\label{eq:error_update_lemma}
e(t+1) = a(t)~e(t) + \left(w_{s,\psi(t)}(t) +w_{p,\psi(t) }(t)\right),
\end{align}
where 
\begin{align}
a(t) = 
\begin{cases}
a_c  \mbox{ if }~\psi(t) \neq \left\{ \emptyset  \right\}\notag\\
a  \mbox{ if }~\psi(t) = \left\{ \emptyset  \right\}.
\end{cases}
\end{align} 
Above, for a packet $\psi$, its ``sensing noise'' $w_{s,\psi}(t)$ and ``process noise'' $w_{p,\psi}(t)$ at time $t$ are defined as follows
\vspace*{-1mm}
\begin{small}
\begin{align}\label{def:noises_lemma}
w_{s,\psi}(t) :& = 
\begin{cases}
 a~K \left(a^{ \tau_{\psi}(t)  } w_s(t-\tau_{\psi}(t)) \right) \mbox{ if }~\psi(t) \ne \left\{ \emptyset  \right\}  \notag\\
 0 \mbox{ if }~\psi(t) = \left\{ \emptyset  \right\} 
\end{cases}
\\
w_{p,\psi }(t)  :& = 
\begin{cases}
-a~K\left(\sum_{l=1}^{\tau_{\psi}(t) } a^{l} w(t - l)\right) + w(t),\mbox{if }~\psi(t) \neq \left\{ \emptyset  \right\}  \notag\\
 w(t)\mbox{ if }~\psi(t) = \left\{ \emptyset  \right\}. 
\end{cases}
\end{align}
\end{small}
We note that the sensing noise depends upon the age of the packet as well as the precision of its source, while the process noise is a function of its age only. 

\emph{Scheduling Problem for $t=T-1$}: Now consider the case when $Q(T-1) \neq \left\{ \emptyset \right\}$, and the scheduler has to make a decision only during the last time-slot $t=T-1$. It follows from the update equation~\eqref{eq:error_update_lemma} that the cost incurred at time $t=T$ is given by 
\begin{align} 
\mathbb{E}~e^2(T) &=  \left(p a^{2}_c + (1-p) a^{2}\right) e^2(T-1) \\
& + p~\mathbb{E} \left(w_{s,\psi(t)}(t) +w_{p,\psi(t) }(t)\right)^2 +  (1-p)  \mathbb{E} w^2(t)\notag\\
&=  \left(p a^{2}_c + (1-p) a^{2}\right) e^2(T-1) + (1-p) \sigma^2\notag \\
&+ p \left( \mathbb{E}\left((w_{s,\psi(t)}(t)  \right)^2 + \mathbb{E}\left((w_{p,\psi(t)}(t)  \right)^2   \right) \notag\\
&=  \left(p a^{2}_c + (1-p) a^{2}\right) e^2(T-1) + (1-p) \sigma^2\notag \\
&+ p\left(  a^{ 2\tau_{\psi}(T)} \sigma^2_{s,\psi} + \sigma^2 \frac{a^{2\tau_{\psi}(T)} -1}{a^2 -1}\right) \notag \\
& + p(1+ K)^2 \sigma^2,
\end{align}
where the second equality follows since the process and sensing noise are independent. Thus, it is optimal to schedule the packet $\psi$ with the least value of $\mathbb{E}\left((w_{s,\psi}(T)  \right)^2 + \mathbb{E}\left((w_{p,\psi}(T)  \right)^2  $.

Thus, we define the following two quantities for a packet $\psi$, which correspond to the sensing and process noise variances associated with it at time $t$,
\begin{align}
W^{2}_{s,\psi}(t) : &=  a^2K^2\left(a^{ 2\tau_{\psi}(t)} \sigma^2_{s,\psi} \right), \notag\\
W^{2}_{p,\psi}(t) :&=  a^2K^2 \left(\sigma^2 \frac{a^{2\tau_{\psi}(t)} -1}{a^2 -1} \right)+ \sigma^2, \notag \\
W^{2}_{\psi}(t) :&= W^{2}_{s,\psi}(t) + W^{2}_{p,\psi}(t).\label{def:noises_variance_pkt}
\end{align}
For an empty packet, i.e., $\psi = \left\{\emptyset \right\}$, we let
\begin{align}
W^{2}_{s,\psi}(t) : &=  0, \notag\\
W^{2}_{p,\psi }(t)  : &= \sigma^2,\label{def:noises_variance_pkt_1}
\end{align}
for all $t \in [1,T]$.

It then follows 
that the policy which chooses the packet with the least value of the combined noise variance $W^{2}_{s,\psi}(T)+W^{2}_{p,\psi}(T)$, is optimal. Thus, we introduce the \emph{Index Policy}.
\begin{definition}(Index Policy)\label{def:index}
The Index Policy assigns an index of $W^2_{\psi}(t)$ to each packet $\psi$ in $Q(t)$, and transmits the packet with the least index. Thus, the scheduling decision $u(t)$ taken at time $t$ is
\begin{align}\label{eq:greedy_scheduler}
u(t) \in \arg\min \left\{ W^2_{\psi} (t+1)\big | \psi \in Q(t)       \right\}.
\end{align}
\end{definition}
\begin{lemma}\label{lemma:index_opt}
The Index Policy of Definition~\ref{def:index} is optimal at time $t=T-1$.
\end{lemma}

Next, we derive some properties of the indices $W^2_{\psi}(t)$ which will allow for efficient implementation of the Index Policy. 

\begin{definition}
Henceforth we let $Q(t)$ denote the \emph{ordered} set of packets in the queue $Q(t)$, with the ordering between two packets $\psi_1,\psi_2 \in Q(t)$ given as 
\begin{align*}
\psi_1 \leq_{Q(t)} \psi_2 \iff W^{2}_{\psi_1}(t)\leq W^{2}_{\psi_2}(t).
\end{align*} 
For two packets $\psi_a,\psi_b$ we let $\Delta(\psi_a,\psi_b,t)$ denote the difference in the values of indices of the two packets, i.e., 
\begin{align}\label{eq:delta_def}
\Delta(\psi_a,\psi_b,t) := W^2_{\psi_a}(t) - W^2_{\psi_b}(t).
\end{align}
\end{definition}

\begin{lemma}\label{lemma:index_order}
Consider two packets $\psi_1,\psi_2 \in Q(t_1)\cap Q(t_2)$, i.e., they are present with the scheduler in its queue at times $t_1$ as well as at time $t_2$. We then have that 
\begin{align*}
\psi_1 \leq_{Q(t_1) } \psi_2 \iff \psi_1 \leq_{Q(t_2) } \psi_2.
\end{align*}
Moreover, 
\begin{align}
 \Delta(\psi_1,\psi_2,t) = a^{2t} C(\psi_1,\psi_2),
\end{align} 
where $C(\psi_1,\psi_2)$ depends upon the packet generation times, and the sensor precisions, of the packets $\psi_1,\psi_2$.
\end{lemma}
\begin{proof}
To show that the ordering of packets doesn't change with time, it suffices to prove that the following holds for times $t_1,t_2$, 
\begin{align*}
W^{2}_{\psi_1}(t_1)\leq W^{2}_{\psi_2}(t_1)  \iff W^{2}_{\psi_1}(t_2)\leq W^{2}_{\psi_2}(t_2).
\end{align*} 
The index of a packet at time $t$ can be derived as a function of its generation time $t_{\psi}$ as follows,
\begin{align*}
W^{2}_{\psi}(t) &=  a^2K^2   \left(a^{ 2( t- t_{\psi} ) } \sigma^2_{s,\psi}+  \sigma^2 \frac{a^{2 (t - t_{\psi})} -1}{a^2 -1}\right) +  \sigma^2\\
&= a^2K^2\left( a^{ 2( t- t_{\psi} ) } \left(\sigma^2_{s,\psi}+  \sigma^2 \frac{1}{a^2 -1} \right) - \sigma^2 \frac{1}{a^2 -1} \right) \\
&+ \sigma^2.
\end{align*}  
Thus, in order to compare the indices of $\psi_1,\psi_2$ at time $t$, we need to solve the following inequality
\begin{align}
&a^{ 2( t- t_{\psi_1} ) } \left(\sigma^2_{s,\psi_1}+  \sigma^2 \frac{1}{a^2 -1} \right) \notag \\
& \leq a^{ 2( t- t_{\psi_2} ) } \left(\sigma^2_{s,\psi_2}+  \sigma^2 \frac{1}{a^2 -1} \right),\notag\\
\equiv & a^{ 2(- t_{\psi_1} ) } \left(\sigma^2_{s,\psi_1}+  \sigma^2 \frac{1}{a^2 -1} \right) \notag \\
& \leq a^{ 2(- t_{\psi_2} ) } \left(\sigma^2_{s,\psi_2}+  \sigma^2 \frac{1}{a^2 -1} \right).\label{ineq:1}
\end{align}
Since the terms in the above inequality do not depend upon time $t$, it then follows that either $\psi_1 \geq_{Q(t)} \psi_2 \forall t$ or vice--versa. This proves the first claim. The second claim follows from the inequality~\eqref{ineq:1} after letting $C(\psi_1,\psi_2)$ to be equal to the difference between the terms on the l.h.s. and r.h.s. 
\end{proof}

It follows from Lemma~\ref{lemma:index_order} that the Index Policy can be implemented as follows. The network node maintains a single ordered queue $Q(t)$ in which it stores packets from \emph{all} the sources, i.e., the queue is \emph{shared} amongst all the sources. Upon receiving a  new packet, it calculates its index, and then enqueues it in a location so that the resulting queue is ordered. Such an operation requires $O(\log|Q(t)|)$ computation if the queue is implemented using a binary search tree. Moreover, it does not require the scheduler to compute the indices of all the packets in $Q(t)$. 

Moreover, it follows from Lemma~\ref{lemma:index_order} that since the order of packets does not change with $t$, the Index Policy can base its decisions on the values $W^{2}_{\psi}(t)$ instead of $W^{2}_{\psi}(t+1)$. Hence, the Index Policy is equivalently given by 
\begin{align}\label{eq:greedy_scheduler_1}
u(t) \in \arg\min \left\{ W^2_{\psi} (t)\big | \psi \in Q(t)       \right\}.
\end{align}
In later discussions, we will switch between the two definitions depending on which is more convenient.
\section{Optimality of Index Policy}
The Index Policy of the previous section is very simple to implement and is optimal when the scheduler has to make decisions for only the last time--slot $t=T-1$. Surprisingly, as will be shown in this section, it turns out that the policy continues to be optimal when implemented over the interval $[1,T]$. 
 
In order that the problem is tractable, we make the following assumption about the parameter $a$ of the process $x(t)$, and the estimator gain $K$ of the filter~\eqref{eq:kalman_filter}.
\begin{assumption}\label{assum:1}
The quantities $a,a_c,K$ satisfy
\begin{align*}
|a|,|a_c| < 1,
\end{align*}
i.e., the process that is being monitored, as well as the filter, are stable. Additionally, we require the 
estimator gain $K$ to be bounded as follows,
\begin{align*}
K\le \frac{1-a^2}{a^2}.
\end{align*}
\end{assumption}  
 
Denote the index policy by $\pi_{idx}$. Let $\tilde{\pi}$ be a policy that uses a different decision rule for time $t=1$, and then employs $\pi_{idx}$ for the remaining time-slots $t\in \left[2,T-1\right]$. Our approach to proving the optimality of $\pi_{idx}$ will be to show that the cost incurred by $\tilde{\pi}$ is more than that of $\pi_{idx}$.

Let $\psi_1$ denote a 
the packet with the least index in $Q(1)$, and let $\psi_2$ be a packet satisfying $\psi_2 >_{Q(1)} \psi_1$. At time $t=1$, $\pi_{idx}$ serves $\psi_1$ while $\tilde{\pi}$ serves $\psi_2$. With a probability equal to $1-p$, the transmission at time $t=1$ is unsuccesful, so that the queues at time $t=2$ are the same for both the policies. Moreover since $\pi_{idx}$ and $\tilde{\pi}$ agree on decisions at times $t\in [2,T-1]$, this implies that 
with probability $1-p$ their cumulative costs in the interval $[1,T]$ are equal. Hence, without loss of generality, we only consider the case when the first packet transmission is successful. 

Define the following stopping times:
 \begin{align}\label{eq:stop}
T_{1} :&= \left\{ t: \psi(t;\pi_{idx}) = \psi_2 \right\},\notag\\
T_2 :&= \left\{ t: \psi(t;\tilde{\pi}) = \psi_1 \right\},
\end{align}
where the notation $\psi(\cdot;\pi )$ explicitly shows the dependence of the scheduling decisions on the policy $\pi$. Thus, $T_1$ is the time when $\pi_{idx}$ delivers $\psi_2$ to the destination.

The following result is easily derived, and is proven in the Appendix.
\begin{lemma}\label{lemma:compare_times}
Consider the stopping times $T_1,T_2$ as defined in~\eqref{eq:stop}. We have
\begin{align*}
T_2 \leq T_1.
\end{align*}
\end{lemma}
Now define $\pi_{idx} \left[T_2, T_{1} -1  \right] $ be the \emph{ordered} set of packets that were delivered under $\pi_{idx}$ in the time interval $\left[T_2, T_{1} -1 \right]$. Similarly consider the set $\tilde{\pi}\left[T_{2} +1 , T_{1}  \right] $. We will now show that these sets are equal.
\begin{lemma}\label{lemma:ordered_sets_eq}
\begin{align*}
\pi_{idx} \left[T_2, T_{1} -1  \right] = \tilde{\pi}\left[T_{2} +1 , T_{1}  \right].
\end{align*}
\end{lemma}
 \begin{proof}(Proof of Lemma~\ref{lemma:ordered_sets_eq})
 
Since the two policies are the same for times $t>1$, and since both of them have the same set of packets at each time, except possibly packet $\psi_1$ or $\psi_2$, we have
\begin{align}\label{seq:2}
\pi_{idx} \left[2, T_{2} -1  \right] = \tilde{\pi}\left[2 , T_{2} -1 \right].
\end{align} 
It then follows from~\eqref{seq:2} that  
\begin{align}
\pi_{idx} \left[1, T_{1}  \right] &= \psi_1~ \pi_{idx} \left[2, T_{2} -1 \right]~\pi_{idx} \left[T_{2},T_{1} -1  \right]~\psi_2\label{seq:1}\\
\tilde{\pi} \left[1, T_{1}  \right] &= \psi_2~ \pi_{idx} \left[2, T_{2} -1 \right]~\psi_1 ~ \tilde{\pi} \left[T_{2}+1,T_{1}  \right]. \label{seq:3}
\end{align}
Since the sets $ \pi_{idx} \left[1, T_{1}  \right], \tilde{\pi}\left[1 , T_{1}  \right]$ consist of the same elements, it follows from~\eqref{seq:2} that the sets $\pi_{idx} \left[T_{2},T_{1} -1  \right]$ and $\tilde{\pi} \left[T_{2}+1,T_{1}  \right]$ consist of the same elements. Moreover, since under either of the two policies, the packets in both these sets are served in the same order, this shows that 
\begin{align}\label{seq:5}
\pi_{idx} \left[T_{2},T_{1} -1  \right] =\tilde{\pi} \left[T_{2}+1,T_{1}  \right].
\end{align}
 \end{proof} 
Throughout, we let 
\begin{align}
\alpha_{s,t}:= \prod_{m=s}^{t} a^2(t).
\end{align} 
 
\begin{theorem}\label{th:index_opt}
Consider the problem of scheduling data packets in order to minimize the estimation error $\mathbb{E}\left( \sum_{t=1}^T  e^2(t)\right)$. Let $\tilde{\pi}$ be a policy that differs from $\pi_{idx}$ on scheduling decision for time $t=1$, and then employs $\pi_{idx}$ for the remaining time-slots $t\in \left[2,T-1\right]$. Then we have that 
\begin{align*}
\mathbb{E}_{\tilde{\pi}}\left( \sum_{t=1}^T  e^2(t)\right) \geq \mathbb{E}_{\pi_{idx}}\left( \sum_{t=1}^T  e^2(t)\right).
\end{align*}
\end{theorem}
\begin{proof}
Consider a policy $\tilde{\pi}_{idx}$ that follows $\pi_{idx}$ except that it serves $\psi_2$ at times when $\tilde{\pi}$ serves $\psi_1$. We will show that $\tilde{\pi}_{idx}$ attains a lower cost than $\tilde{\pi}$. The \emph{ordered} sets of packets served by them until the time $T_1$ are given by,
\begin{align}
\pi_{idx} \left[1, T_{1}  \right] &= \psi_1~ \pi_{idx} \left[2, T_{2} -1 \right]~\pi_{idx} \left[T_{2},T_{1} -1  \right]~\psi_2\label{seq:idx}\\
\tilde{\pi}_{idx} \left[1, T_{1}  \right] &= \psi_1~ \pi_{idx} \left[2, T_{2} -1 \right]~\psi_2~\pi_{idx} \left[T_{2}+1,T_{1} -1  \right].\label{seq:tilde_idx}\\
\tilde{\pi} \left[1, T_{1}  \right] &= \psi_2~ \pi_{idx} \left[2, T_{2} -1 \right]~\psi_1 ~ \pi_{idx} \left[T_{2}+1,T_{1} -1  \right]. 
\end{align}
We note that since all the policies are constructed on the same probability space, and since they are non-idling, the $a(t)$ are the same under all policies. The performance difference between $\tilde{\pi}_{idx},\tilde{\pi}$ is given by
\begin{small}
\begin{align}
&\left(\sum_{t=1}^{T} e^2(t) \right)_{\tilde{\pi}_{idx}}-  \left(\sum_{t=1}^{T} e^2(t) \right)_{\tilde{\pi}} \notag \\
 & =  \Delta(\psi_1,\psi_2,1)  \sum_{t\in [1,T_2 -1]}\alpha_{1,t} + \big(\Delta(\psi_2,\psi_1,T_2) \notag \\
 & + \Delta(\psi_1,\psi_2,1)\alpha_{1,T_2} \big) \sum_{t\geq T_2}\alpha_{T_2,t}\notag\\
&=  \Delta(\psi_1,\psi_2,1)  \sum_{t\in [1,T_2 -1]}\alpha_{1,t} + \big( (a^2)^{T_2}\Delta(\psi_2,\psi_1,1) \notag \\
& + \Delta(\psi_1,\psi_2,1)\alpha_{1,T_2} \big) \sum_{t\geq T_2}\alpha_{T_2,t}\notag\\
&=  \Delta(\psi_1,\psi_2,1) \left( \sum_{t\in [1,T_2 -1]}\alpha_{1,t} - \left( (a^2)^{T_2} -  \alpha_{1,T_2} \right) \sum_{t\geq T_2}\alpha_{T_2,t} \right)\notag\\
&\leq 0,\label{ineq:pol_1}
\end{align} 
\end{small}
\noindent where, in the first equality we used Lemma~\ref{lemma:index_order} in order to deduce $\Delta(\psi_2,\psi_1,1) = (a^2)^{T_2}\Delta(\psi_2,\psi_1,1)$. To show the last inequality, we prove in Lemma~\ref{lemma:negative_delta} that the term within the braces is positive, and also note that $\Delta(\psi_1,\psi_2,1)<0$.

We will now show that $\pi_{idx}$ has a lower cost than $\tilde{\pi}_{idx}$. Towards this end, construct a class of feasible policies $\tilde{\pi}_{idx,k},k\geq 0$. Denote $\tilde{\pi}_{idx}$ by $\tilde{\pi}_{idx,0}$. $\tilde{\pi}_{idx,k}$ is obtained from $\tilde{\pi}_{idx,k-1}$ as follows. $\tilde{\pi}_{idx,k}$ does not serve $\psi_2$ until $\tilde{\pi}_{idx,k-1}$ has delivered $\psi_2$. However, in case $\pi_{idx}$ attempts $\psi_2$, then $\tilde{\pi}_{idx,k}$ also attempts it. 
Using the same techniques that were used for proving the inequality~\eqref{ineq:pol_1}, it can be shown that the costs satisfy 
\begin{align*}
\left(\sum_{t=1}^{T} e^2(t) \right)_{\tilde{\pi}_{idx,k}}\leq  \left(\sum_{t=1}^{T} e^2(t) \right)_{\tilde{\pi}_{idx,k-1} }, k\in \mathbb{Z}^{+}. 
\end{align*} 
Since it was shown in~\eqref{ineq:pol_1} that $\tilde{\pi}_{idx}$, equivalently $\tilde{\pi}_{idx,0}$, has a lower cost than $\tilde{\pi}$, this shows that
\begin{align*}
\left(\sum_{t=1}^{T} e^2(t) \right)_{\tilde{\pi}_{idx,k} } \leq \left(\sum_{t=1}^{T} e^2(t) \right)_{\tilde{\pi}_{idx} } ,~k\in \mathbb{Z}^{+}. 
\end{align*}
Now, observe that the cost incurred by $\pi_{idx}$ is equal to that incurred by $\tilde{\pi}_{idx,k}$ with a probability equal to 
\begin{align*}
\mathbb{P}\left(  |\pi_{idx} \left[T_{2}+1,T_{1} -1  \right]|=k | \mathcal{F}_{T_2} \right).
\end{align*} 
Since the costs of $\tilde{\pi}_{idx,k}$ are lower than that of $\tilde{\pi}$, we then have that
\begin{align*}
\left(\sum_{t=1}^{T} e^2(t) \right)_{\pi_{idx} } \leq \left(\sum_{t=1}^{T} e^2(t) \right)_{\tilde{\pi}}.
\end{align*}
This completes the proof.
\end{proof} 
Consequently, we are now in a position to state the optimality of the Index Policy.
\begin{corollary}
The Index Policy that serves according to 
\begin{align}
u(t) \in \arg\min \left\{ W^2_{\psi}(t) \big | \psi \in Q(t)       \right\},
\end{align}
minimizes the estimation error $\mathbb{E}\left( \sum_{t=1}^T  e^2(t)\right)$ under the Assumption~\ref{assum:1}. The optimal cost is a function of the starting state $e(0),Q(0)$, the time-duration that the system operates, and the channel reliability. Henceforth we denote this optimal cost by $J^{\star}(e(0);p;T)$.
\end{corollary}
\begin{proof}
Since the problem~\eqref{eq:cost} is a Markov decision process (MDP)~\cite{puterman}, we can obtain an optimal policy by solving the corresponding dynamic programming recursions. Let us assume that Index Policy is optimal when the scheduler has to make decisions for times $t\in [s,T]$, or equivalently the DP backward induction until time $t=s$ yields us $\pi_{idx}$. It then follows from Theorem~\ref{th:index_opt} that Index Policy continues to be optimal at time $t=s-1$, i.e. the DP recursion for time $t=s-1$ also yields $\pi_{idx}$. Since it was shown in Lemma~\ref{lemma:index_opt} that Index Policy is optimal at time at time $t=T-1$, we conclude by using backward induction that it is also optimal for all times $t\in [1,T]$. 
\end{proof}

\section{Multiple Processes Sharing a Common Network For Communicating Update Packets}\label{sec:multiple_users}
In the setup considered so far, there is only a single process of interest. However, modern cyber-physical system (CPS) such as the smart-grid, vehicular networks, etc. are comprised of multiple complex processes that need to be monitored in real-time by their respective destination nodes. 

Thus, we now consider the case when multiple processes share the same network in order to transport their packets containing updates to their respective destination nodes. 
As shown in Fig.~\ref{fig:decomposition_1} there are $N$ processes denoted $x_i(t),i\in [1,N]$. The evolution of the $i$-th process $x_i(t)$ is described as
\begin{align}\label{eq:n_dyn}
x_i(t+1) = a_i x_i(t) + w_i(t), t\in [1,T-1],
\end{align}
where the noise process $w_i(t)$ is i.i.d. across times for each of the $N$ processes, has a Gaussian distribution with $0$ mean with variance $\sigma^2_{i,p}$. The $w_{i}(t)$ are also independent across different processes. 

Each process is monitored by multiple sources. For sake of simplicity, we assume that there are $M$ sources for each of the $N$ processes, though the results can easily be generalized to the case when the numbers of sources differ across processes. The sources share a single network node $\mathcal{N}$ that transmits their packets to their respective destination nodes $\mathcal{D}_{i},i\in [1,N]$ via an unreliable link that has a reliability $p$. $\mathcal{N}$ maintains $N$ separate queues $Q_i(t),i\in [1,N]$, and stores packets from the sources of process $i$ in $Q_i(t)$. The destination node $\mathcal{D}_i$ maintains a Kalman-like linear filter for estimating $x_i(t)$, and its estimate is denoted $\hat{x}_i(t)$ which is updated as in~\eqref{eq:kalman_filter} with the term $\hat{s}(t)$ replaced by $\hat{x}_i(t)$.

The network node $\mathcal{N}$ makes scheduling decisions $u(t),t\in [1,T-1]$ regarding which packet from the $N$ queues, i.e., $\cup_{i=1}^{N} Q_i(t)$, should be scheduled for transmission. It can utilize its past observations, i.e.,  $\left\{Q(s)\right\}_{s=1}^{t},\left\{u(s),x_{\ell}(s)\right\}_{s=1}^{t-1}$, where $x_{\ell}(s)$ is the state of the channel which connects $\mathcal{N}$ to the destination node of the queue which was served at time $s$. More concretely, it can utilize a) the age values and the sensor precision information of all the packets it has received so far  b) the channel state value for the channels it used for past transmissions. 

The objective to be minimized is the cumulative quadratic estimation error of $N$ users during the time period $[1,T]$. Thus,
the \emph{Multi Process Scheduling Problem} is:
\begin{align}\label{eq:problem_n}
&\min_{\pi}J(\pi) :=\mathbb{E}_{\pi}\sum_{i=1}^{N}\left(\sum_{t=1}^{T} e^2_{i}(t)\right),
\end{align}
where $e_i(t)$ is the estimation error for $i$-th process.

It can be shown, using standard results on controlled Markov processes~\cite{kumar,puterman}, that it suffices to consider only Markovian policies, i.e., those policies for which $u(t)$ is a function of $\left\{Q_i(t)\right\}_{i=1}^{N}$. However, to derive the optimal policy, we need to obtain the value function associated with the stochastic control problem~\eqref{eq:problem_n} by solving the Dynamic Programming recursions~\cite{kumar,puterman}. For the problem~\eqref{eq:problem_n} the system state at time $t$ is described by $\left\{Q_i(t)\right\}_{i=1}^{N}$, i.e., the age and sensor precision values for all the packets present in the $N$ queues. The state-space thus grows exponentially with $N$. Since the computational complexity of Dynamic Programming recursions is proportional to the size of state-space, it is impractical to solve for the optimal policy. Thus, we now restrict the optimization problem~\eqref{eq:problem_n} to a particular class of policies, and this allows us to obtain a tractable solution.

Let the scheduling decision at time $t$ be $u(t)=\left(u_1(t),u_2(t)\right)$,  where $u_1(t)\in [1,N]$ decides which of the $N$ queues is served, while $u_2(t)$ denotes the packet from $Q_{u_1(t)}(t)$ that is scheduled for transmission. We will restrict ourselves to the class of policies under which the $u_1(t)$'s are i.i.d. across time and do not depend upon system state $\left\{Q_i(t) \right\}_{i\in[1,N]}$. Such a policy is parametrized by a probability vector $\vec{p}:=\left\{ p_i \right\}_{i\in [1,N]}$ and $u_1(t)=i$ w.p. $p_i$ at each time. Thus, $\vec{p}\in \Delta(N)$, where $\Delta(N)$ is the $N$-dimensional simplex. We denote the class of such policies by $\Pi_{i.i.d.}$, i.e.,
\begin{align*}
\Pi_{i.i.d.}  := \big\{ \pi: & \mathbb{P}\left(u_1(t) =i | \left\{Q_i(s)\right\}_{i=1}^{N}, s\in [1,t] \right) = p_i, \\ 
&\forall t, i\in [1,N]  \big\}.
\end{align*}
Fix a $\vec{p}\in\Delta(N)$ and let $u_1(t)$ be according to $\vec{p}$. We will now derive a scheduling algorithm that will prioritize the packets on basis of their ages and sensor qualities, thereby making the decisions $u_2(t)$, i.e., we will solve
\begin{align}
&\min_{\left\{u_2(t)\right\}_{t\in [1,T]}}J :=\mathbb{E}\sum_{i=1}^{N}\left(\sum_{t=1}^{T} e^2_{i}(t)\right)\label{eq:problem_n_iid}\\
&\mbox{ s.t. } u_1(t) \sim \vec{p}~\mbox{ i.i.d. } \label{eq:problem_n_iid_const}
\end{align}

Now, since the decisions $u_1(t)$ are chosen i.i.d. according to $\vec{p}$, in order to analyze the combined system comprising of $N$ processes, we could equivalently assume that we are dealing with $N$ systems, where each system is composed of a \emph{single process} that is being estimated remotely. At the beginning of time-slot $t$, the queue for process $i$ obtains access to the channel for transmitting packet The packet transmission, if any, is successful w.p. $p$. Thus, the scheduler for the $i$-th process has to make scheduling decisions regarding the packets of \emph{only} process $i$, and can base them on the knowledge of only $Q_i(t)$ and not $\left\{Q_i(t)\right\}_{i\in [1,N] }$.

This decomposition property is shown in Fig.~\ref{fig:decomposition_2}. Next, we make this decomposition property concrete.
\begin{figure*}[!h]
\begin{subfigure}{\textwidth}
	\centering
	\includegraphics[width=0.5\textwidth]{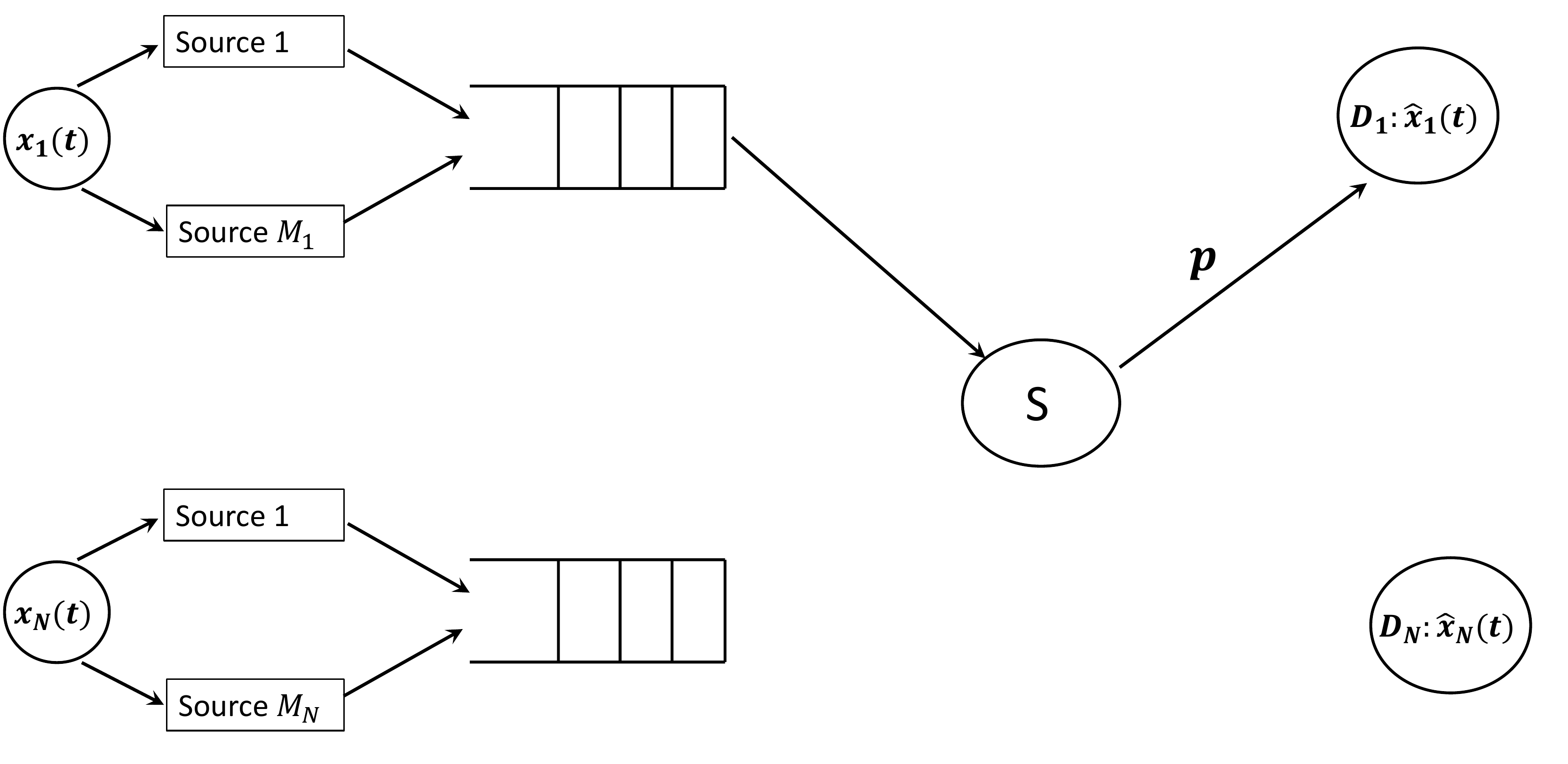}
	\caption{The original problem requires making decisions for $N$ processes, and suffers from the curse of dimensionality since the computational complexity of obtaining optimal policy grows exponentially with $N$.}
	\label{fig:decomposition_1}
	\end{subfigure}
\begin{subfigure}{\textwidth}
	\centering
	\includegraphics[width=0.5\textwidth]{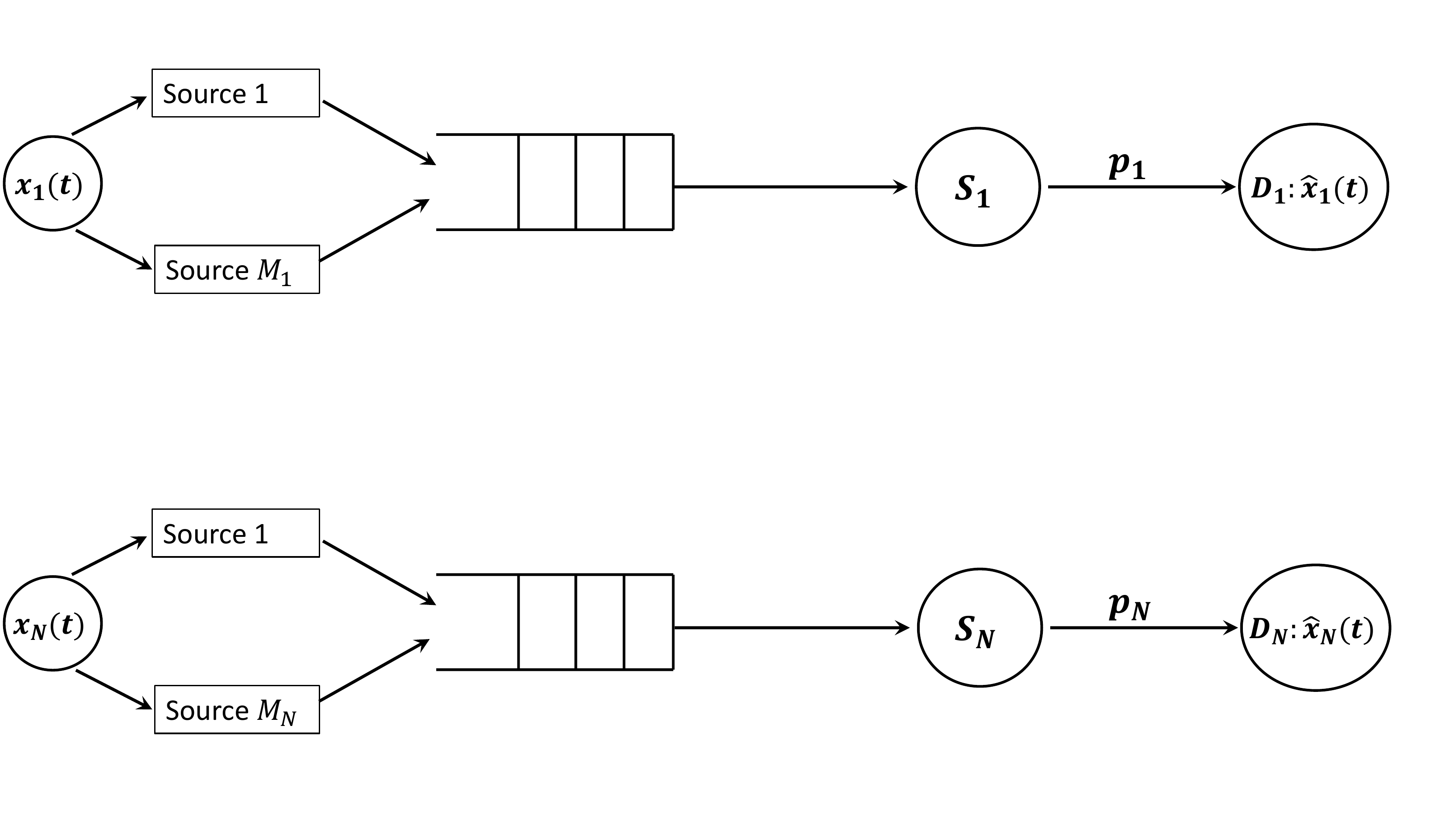}
	\caption{The problem of Fig.~\ref{fig:decomposition_1} decomposes into $N$ single-process problems, so that its complexity grows linearly with $N$.}
	\label{fig:decomposition_2}
	\end{subfigure}
\end{figure*}
Consider the problem~\eqref{eq:problem_n_iid}-\eqref{eq:problem_n_iid_const} of optimally scheduling data packets for $N$ processes while restricting to the class $\Pi_{i.i.d.}$. We will obtain a solution to it by solving $N$ \emph{single-process problems}. We define the single-process problems below.
\begin{definition}(Single Process Scheduling Problem)
The optimal scheduling problem for the $i$-th process is described as follows. There is a single Gauss-Markov process $x_i(t)$ that evolves as~\eqref{eq:n_dyn}. $M$ sensors record observations of the process $x_i(t)$, encode it into data packets and send them to the network node $\mathcal{N}_i$. $\mathcal{N}_i$ is connected to the destination node $\mathcal{D}_i$ via a channel of reliability $p_i$. The node $\mathcal{N}_i$ gets the opportunity to transmit if $u_1(t)=i$. The process $u_i(t)\in [1,N]$ is i.i.d. and distribted according to $\vec{p}$. 
The scheduler at $\mathcal{N}_i$ is denoted $\pi_i$, and has to make scheduling decisions $u_2(t)$ on the basis of $Q_i(t),u_1(t)$ in order to minimize the estimation error
\begin{align*}
\mathbb{E}_{\pi_i} \left(\sum_{t=1}^{T} e^2_i(t) \right).
\end{align*}
The Index Policy is optimal for the above problem, and attains a cost equal to $J^{\star}_i(e_i(0);T)$. We denote the index policy applied to solve the above problem for process $i$ by $\pi^{(i)}_{idx}$.
\end{definition}
Next, we obtain a scheduler for the original problem~\eqref{eq:problem_n_iid}-\eqref{eq:problem_n_iid_const} by combining solutions of the \emph{single-process problems}. We denote by $\otimes_{i=1}^{N} \pi^{(i)}_{idx}$ the scheduler which implements the policy $\pi^{(u_1(t))}_{idx}$ at time $t$.

\begin{theorem}\label{lemma:separable_opt_pol}
Consider the problem~\eqref{eq:problem_n_iid}-\eqref{eq:problem_n_iid_const} of scheduling data packets for $N$ processes, where the channel access decisions are made in an i.i.d. manner. The policy $\otimes_{i=1}^{N} \pi^{(i)}_{idx}$ is optimal, and its estimation error cost is equal to $\sum_{i=1}^{N} J^{\star}_i(e_i(0);pp_i;T)$.
\end{theorem}

We can further decrease the cost~\eqref{eq:problem_n} by optimizing over the $u_1(t)$ proess, i.e., the choice of vector $\vec{p}$. One possibility to perform this is to use an ``online learning" algorithm such as the Simultaneous Perturbation Stochastic Approximation (SPSA)~\cite{spall1992multivariate}.

\section{Conclusion}\label{sec:conclusion}
We considered the problem of optimally scheduling pkts. in order to minimize the estimation error of a remote estimator that relies upon data pkts. that it receives from multiple sources.
We have shown that the optimal policy prioritizes pkts. on the basis of their \emph{Value of Information}. The Value of Information contained in a pkt. depends not only upon its age, i.e., how stale is the information content, but also on the sensor precision, i.e., the quality of information. Such a policy is easily implementable.

Further investigations include important extensions of the results that we obtained in this work. These include joint optimization of the filtering and the scheduling processes, decentralized scheduling over multihop networks, vector-valued processes, bounding the sub-optimality arising from using a constant gain, etc.

\section{Appendix}
\subsection{Proofs of Lemmas}
\begin{proof}[Lemma~\ref{lemma:index_opt}]
Since,
\begin{align*}
 a^{ \tau_{\psi}(s)  }y_{\psi} &=  a^{ \tau_{\psi}(s) } x(t - \tau_{\psi}(s) ) + a^{ \tau_{\psi}(s)  } w_s(t-\tau_{\psi}(s)),\\
\mbox{ and } ~~x(s) &= a^{ \tau_{\psi}(s) } x(s - \tau_{\psi}(s) ) + \sum_{l=1}^{\tau_{\psi}(s)} a^{l} w(s - l),
\end{align*}
we have
\begin{align} \label{eq:weighted_measure}
  a^{ \tau_{\psi}(s)  }y_{\psi} = x(s)- \sum_{l=1}^{\tau_{\psi}(s) } a^{l} w(s - l) + a^{ \tau_{\psi}(s)  } w_s(t-\tau_{\psi}(s)).
\end{align}

Recall that the filter update equation is given by
\begin{align}
\hat{x}(t+1) &= a \hat{x}(t) + K \left(  a^{ \tau_{\psi(t)}(t+1) }y_{\psi(t)} - a \hat{x}(t) \right)\notag\\
&= a \hat{x}(t) + a K \left(  a^{ \tau_{\psi(t)}(t) }y_{\psi(t)} -  \hat{x}(t) \right)\label{eq:kalman_filter_repeat}
\end{align}
Substituting~\eqref{eq:weighted_measure} into~\eqref{eq:kalman_filter_repeat} with $s=t$, and $\psi = \psi(t)$ we obtain
\begin{small}
\begin{align*}
\hat{x}(t+1) = a \hat{x}(t) + a K\left(  x(t) - \hat{x}(t)\right) + w_{s,\psi(t)}(t) + w_{p,\psi(t)}(t) ,
\end{align*}
\end{small}
\noindent where $w_{s,\psi(t)}(t) + w_{p,\psi(t)}(t)$ are the sensing and process noise as defined earlier. Noting that $x(t+1)=a x(t)+w(t)$, the update for the error $e(t)$ is obtained by subtracting $\hat{x}(t+1)$ from $x(t+1)$, and is equal to
\begin{align}\label{eq:error_update}
e(t+1) &= a (1-K)e(t) - \left(w_{s,\psi(t)}(t) + w_{p,\psi(t)}(t)\right) + w(t).
\end{align}

Now define
\begin{align}\label{def:noises_variance}
W^{2}_{s,\psi}(t) : &=  a^{ 2\tau_{\psi}(t)  } \sigma^2_{s,\psi} \notag\\
W^{2}_{p,\psi }(t)  : &=  \sigma^2 \frac{a^{2\tau_{\psi}(t)} -1}{a^2 -1} + (1+ K)^2 \sigma^2.
\end{align}
\end{proof}

\begin{proof}[Lemma~\ref{lemma:compare_times}]
We note that the index of $\psi_1$ was lesser than that of $\psi_2$ at time $t=1$. Since it was shown in Lemma~\ref{lemma:index_order} that the order of the indices of packets does not change with time, this implies that the index of $\psi_2$ under $\pi_{idx}$ is always greater than the index of $\psi_1$ under $\tilde{\pi}$. Since both $\tilde{\pi},\pi_{idx}$ implement the least index first rule, and since the set of packets received by both policies are the same, it then follows that $\psi_1$ will be served under $\tilde{\pi}$ earlier than the time at which $\psi_2$ will be served under $\pi_{idx}$. This completes the proof.
\end{proof}

\subsection{Some Useful Results}
\begin{lemma}\label{lemma:negative_delta}
Let the system dynamics $a$, and the estimator gain $K$ satisfy the Assumption~\ref{assum:1}, and let  $\alpha_{s,t}:= \prod_{m=s}^{t} a^2(t)$.

 Let $\mathcal{T}$ be a stopping time that satisfies $\mathcal{T}>1$. We then have that 
\begin{align*}
 \left( \sum_{t\in [1,\mathcal{T}]}\alpha_{1,t} - \left( (a^2)^{\mathcal{T}} -  \alpha_{1,\mathcal{T}} \right) \sum_{t > \mathcal{T}}\alpha_{\mathcal{T},t} \right) >0.
\end{align*}
\end{lemma}
\begin{proof}
If the time $\mathcal{T}>1$, we have
\begin{align*}
& \sum_{t\in [1,\mathcal{T}]}\alpha_{1,t} - \left( (a^2)^{\mathcal{T}} -  \alpha_{1,\mathcal{T}} \right) \sum_{t > \mathcal{T}}\alpha_{T,t} \\
&=   \sum_{t\in [1,\mathcal{T}]}\alpha_{1,t} + \alpha_{1,\mathcal{T}} \sum_{t > \mathcal{T}}\alpha_{\mathcal{T},t} - \left( (a^2)^{\mathcal{T}} \right) \sum_{t > \mathcal{T}}\alpha_{\mathcal{T},t} \\
&\geq   \sum_{t\in [1,\mathcal{T}]}  (a^2_c)^{t} +  (a^2_c)^{\mathcal{T}} \sum_{t > \mathcal{T}}\alpha_{\mathcal{T},t} - \left( (a^2)^{\mathcal{T}} \right) \sum_{t > \mathcal{T}}\alpha_{\mathcal{T},t} \\
&\geq   \left(\sum_{t\in [1,\mathcal{T}]}  a^{2t}_c \right) - \frac{ (a^2)^{\mathcal{T}} -  (a^2_c)^{\mathcal{T}} }{1-a^2} \\
&\geq   1 +  a^2_c - \frac{ (a^2)^{\mathcal{T}} -  (a^2_c)^{\mathcal{T}} }{1-a^2} \\
&\geq   1 +  a^2_c - \frac{ (a^2)^{\mathcal{T}}  }{1-a^2} \\
&\geq   1 +  a^2_c - \frac{ a^2 }{1-a^2}, 
\end{align*}
where the inequalities follow from the fact that $0<a_c^2<a^2<1$. Hence, it suffices to show that $1 +  a^2_c - a^2 \slash (1-a^2)\ge 0$ under Assumption~\ref{assum:1}. Setting $a_c = a(1-K)$, this condition simplifies to 
\begin{align*}
1 \ge \frac{a^2  - a^2(1-K)^2 -a^4(1-K)^2 }{1- a^2}. 
\end{align*}
A sufficient condition to ensure the above inequality, is to satisfy
\begin{align*}
1 \ge \frac{a^2  - a^2(1-K)^2  }{1- a^2}, 
\end{align*}
or equivalently 
\begin{align*}
1 \ge a^2\frac{(2-K) K }{1- a^2} \equiv  K(2-K) \le (1-a^2)\slash a^2.
\end{align*}
Since $1-K>0$, the condition is satisfied if $K<\frac{1-a^2}{a^2}$.
\end{proof}
\bibliographystyle{IEEEtran}
\bibliography{combinedbib.bib}

\begin{thebibliography}{10}
\providecommand{\url}[1]{#1}
\csname url@samestyle\endcsname
\providecommand{\newblock}{\relax}
\providecommand{\bibinfo}[2]{#2}
\providecommand{\BIBentrySTDinterwordspacing}{\spaceskip=0pt\relax}
\providecommand{\BIBentryALTinterwordstretchfactor}{4}
\providecommand{\BIBentryALTinterwordspacing}{\spaceskip=\fontdimen2\font plus
\BIBentryALTinterwordstretchfactor\fontdimen3\font minus
  \fontdimen4\font\relax}
\providecommand{\BIBforeignlanguage}[2]{{%
\expandafter\ifx\csname l@#1\endcsname\relax
\typeout{** WARNING: IEEEtran.bst: No hyphenation pattern has been}%
\typeout{** loaded for the language `#1'. Using the pattern for}%
\typeout{** the default language instead.}%
\else
\language=\csname l@#1\endcsname
\fi
#2}}
\providecommand{\BIBdecl}{\relax}
\BIBdecl

\bibitem{kaul2012real}
S.~Kaul, R.~Yates, and M.~Gruteser, ``Real-time status: How often should one
  update?'' in \emph{INFOCOM, 2012 Proceedings IEEE}.\hskip 1em plus 0.5em
  minus 0.4em\relax IEEE, 2012, pp. 2731--2735.

\bibitem{DBLP:conf/allerton/KadotaUSM16}
\BIBentryALTinterwordspacing
I.~Kadota, E.~Uysal{-}Biyikoglu, R.~Singh, and E.~Modiano, ``Minimizing the age
  of information in broadcast wireless networks,'' in \emph{54th Annual
  Allerton Conference on Communication, Control, and Computing, Allerton 2016,
  Monticello, IL, USA, September 27-30, 2016}, 2016, pp. 844--851. [Online].
  Available: \url{https://doi.org/10.1109/ALLERTON.2016.7852321}
\BIBentrySTDinterwordspacing

\bibitem{DBLP:journals/ton/KadotaSUSM18}
\BIBentryALTinterwordspacing
I.~Kadota, A.~Sinha, E.~Uysal{-}Biyikoglu, R.~Singh, and E.~Modiano,
  ``Scheduling policies for minimizing age of information in broadcast wireless
  networks,'' \emph{{IEEE/ACM} Trans. Netw.}, vol.~26, no.~6, pp. 2637--2650,
  2018. [Online]. Available: \url{https://doi.org/10.1109/TNET.2018.2873606}
\BIBentrySTDinterwordspacing

\bibitem{DBLP:conf/isit/BedewySS17}
\BIBentryALTinterwordspacing
A.~M. Bedewy, Y.~Sun, and N.~B. Shroff, ``Age-optimal information updates in
  multihop networks,'' in \emph{2017 {IEEE} International Symposium on
  Information Theory, {ISIT} 2017, Aachen, Germany, June 25-30, 2017}, 2017,
  pp. 576--580. [Online]. Available:
  \url{https://doi.org/10.1109/ISIT.2017.8006593}
\BIBentrySTDinterwordspacing

\bibitem{adelberg1995applying}
B.~Adelberg, H.~Garcia-Molina, and B.~Kao, ``Applying update streams in a soft
  real-time database system,'' in \emph{ACM SIGMOD Record}, vol.~24, no.~2,
  1995, pp. 245--256.

\bibitem{Cho00synchronizinga}
J.~Cho and H.~Garcia-molina, ``Synchronizing a database to improve freshness,''
  2000, pp. 117--128.

\bibitem{DBLP:conf/icde/GolabJS09}
\BIBentryALTinterwordspacing
L.~Golab, T.~Johnson, and V.~Shkapenyuk, ``Scheduling updates in a real-time
  stream warehouse,'' in \emph{Proceedings of the 25th International Conference
  on Data Engineering, {ICDE} 2009, March 29 2009 - April 2 2009, Shanghai,
  China}, 2009, pp. 1207--1210. [Online]. Available:
  \url{https://doi.org/10.1109/ICDE.2009.202}
\BIBentrySTDinterwordspacing

\bibitem{2012ISIT-YatesKaul}
R.~D. Yates and S.~Kaul, ``Real-time status updating: Multiple sources,'' in
  \emph{Proc.~IEEE ISIT}, July 2012, pp. 2666--2670.

\bibitem{DBLP:conf/isit/Yates15}
\BIBentryALTinterwordspacing
R.~D. Yates, ``Lazy is timely: Status updates by an energy harvesting source,''
  in \emph{{IEEE} International Symposium on Information Theory, {ISIT} 2015,
  Hong Kong, China, June 14-19, 2015}, 2015, pp. 3008--3012. [Online].
  Available: \url{https://doi.org/10.1109/ISIT.2015.7283009}
\BIBentrySTDinterwordspacing

\bibitem{DBLP:journals/iotj/JiangKZZN19}
\BIBentryALTinterwordspacing
Z.~Jiang, B.~Krishnamachari, X.~Zheng, S.~Zhou, and Z.~Niu, ``Timely status
  update in wireless uplinks: Analytical solutions with asymptotic
  optimality,'' \emph{{IEEE} Internet of Things Journal}, vol.~6, no.~2, pp.
  3885--3898, 2019. [Online]. Available:
  \url{https://doi.org/10.1109/JIOT.2019.2893319}
\BIBentrySTDinterwordspacing

\bibitem{DBLP:conf/infocom/JiangZN019}
\BIBentryALTinterwordspacing
Z.~Jiang, S.~Zhou, Z.~Niu, and Y.~Cheng, ``A unified sampling and scheduling
  approach for status update in multiaccess wireless networks,'' in \emph{2019
  {IEEE} Conference on Computer Communications, {INFOCOM} 2019, Paris, France,
  April 29 - May 2, 2019}, 2019, pp. 208--216. [Online]. Available:
  \url{https://doi.org/10.1109/INFOCOM.2019.8737404}
\BIBentrySTDinterwordspacing

\bibitem{DBLP:conf/isit/SunPU17}
\BIBentryALTinterwordspacing
Y.~Sun, Y.~Polyanskiy, and E.~Uysal{-}Biyikoglu, ``Remote estimation of the
  wiener process over a channel with random delay,'' in \emph{2017 {IEEE}
  International Symposium on Information Theory, {ISIT} 2017, Aachen, Germany,
  June 25-30, 2017}, 2017, pp. 321--325. [Online]. Available:
  \url{https://doi.org/10.1109/ISIT.2017.8006542}
\BIBentrySTDinterwordspacing

\bibitem{cdcdelay}
{Rahul Singh and P.~R. Kumar}, ``Decentralized throughput maximizing policies
  for deadline-constrained wireless networks,'' \emph{IEEE Conference on
  Decision and Control}, 2015.

\bibitem{singh2018throughput}
\BIBentryALTinterwordspacing
R.~Singh and P.~R. Kumar, ``Throughput optimal decentralized scheduling of
  multihop networks with end-to-end deadline constraints: Unreliable links,''
  \emph{{IEEE} Trans. Automat. Contr.}, vol.~64, no.~1, pp. 127--142, 2019.
  [Online]. Available: \url{https://doi.org/10.1109/TAC.2018.2874671}
\BIBentrySTDinterwordspacing

\bibitem{guo2018risk}
X.~Guo, R.~Singh, P.~Kumar, and Z.~Niu, ``A risk-sensitive approach for packet
  inter-delivery time optimization in networked cyber-physical systems,''
  \emph{IEEE/ACM Transactions on Networking (TON)}, vol.~26, no.~4, pp.
  1976--1989, 2018.

\bibitem{guosingh}
\BIBentryALTinterwordspacing
------, ``A high reliability asymptotic approach for packet inter-delivery time
  optimization in cyber-physical systems,'' in \emph{Proceedings of the 16th
  ACM International Symposium on Mobile Ad Hoc Networking and Computing}, ser.
  MobiHoc '15.\hskip 1em plus 0.5em minus 0.4em\relax New York, NY, USA: ACM,
  2015, pp. 197--206. [Online]. Available:
  \url{http://doi.acm.org/10.1145/2746285.2746305}
\BIBentrySTDinterwordspacing

\bibitem{singh2015index}
R.~Singh, X.~Guo, and P.~R. Kumar, ``Index policies for optimal mean-variance
  trade-off of inter-delivery times in real-time sensor networks,'' in
  \emph{Computer Communications (INFOCOM), 2015 IEEE Conference on}.\hskip 1em
  plus 0.5em minus 0.4em\relax IEEE, 2015, pp. 505--512.

\bibitem{singh2015optimizing}
R.~Singh and P.~Kumar, ``Optimizing quality of experience of dynamic video
  streaming over fading wireless networks,'' in \emph{Decision and Control
  (CDC), 2015 IEEE 54th Annual Conference on}.\hskip 1em plus 0.5em minus
  0.4em\relax IEEE, 2015, pp. 7195--7200.

\bibitem{singh2015maxweight}
R.~Singh and A.~Stolyar, ``Maxweight scheduling: Asymptotic behavior of
  unscaled queue-differentials in heavy traffic,'' in \emph{ACM SIGMETRICS
  Performance Evaluation Review}, vol.~43, no.~1.\hskip 1em plus 0.5em minus
  0.4em\relax ACM, 2015, pp. 431--432.

\bibitem{rs}
------, ``Maxweight scheduling: Asymptotic behavior of unscaled
  queue-differentials in heavy traffic,'' in \emph{Proceedings of the 2015 ACM
  SIGMETRICS International Conference on Measurement and Modeling of Computer
  Systems}, ser. SIGMETRICS '15.\hskip 1em plus 0.5em minus 0.4em\relax New
  York, NY, USA: ACM, 2015, pp. 431--432.

\bibitem{atilla1}
R.~Li, A.~Eryilmaz, and B.~Li, ``Throughput-optimal wireless scheduling with
  regulated inter-service times,'' in \emph{INFOCOM, 2013 Proceedings IEEE},
  April 2013, pp. 2616--2624.

\bibitem{atilla2}
B.~Li, R.~Li, and A.~Eryilmaz, ``Heavy-traffic-optimal scheduling with regular
  service guarantees in wireless networks,'' in \emph{Proceedings of the
  Fourteenth ACM International Symposium on Mobile Ad Hoc Networking and
  Computing}, ser. MobiHoc '13.\hskip 1em plus 0.5em minus 0.4em\relax ACM,
  2013, pp. 79--88.

\bibitem{kumar}
{P. R. Kumar and P. Varaiya}, \emph{{Stochastic systems: Estimation,
  identification and adaptive control}}.\hskip 1em plus 0.5em minus 0.4em\relax
  Prentice Hall Inc., Englewood Cliffs, 1986.

\bibitem{DBLP:journals/automatica/ShiXM09}
\BIBentryALTinterwordspacing
L.~Shi, L.~Xie, and R.~M. Murray, ``Kalman filtering over a packet-delaying
  network: {A} probabilistic approach,'' \emph{Automatica}, vol.~45, no.~9, pp.
  2134--2140, 2009. [Online]. Available:
  \url{https://doi.org/10.1016/j.automatica.2009.05.018}
\BIBentrySTDinterwordspacing

\bibitem{puterman}
M.~L. Puterman, \emph{Markov Decision Processes: Discrete Stochastic Dynamic
  Programming}, 1st~ed.\hskip 1em plus 0.5em minus 0.4em\relax New York, NY,
  USA: John Wiley \& Sons, Inc., 1994.

\bibitem{spall1992multivariate}
J.~C. Spall, ``Multivariate stochastic approximation using a simultaneous
  perturbation gradient approximation,'' \emph{IEEE transactions on automatic
  control}, vol.~37, no.~3, pp. 332--341, 1992.

\end{thebibliography}
\end{document}